\documentclass[manuscript,nonacm]{acmart}

\usepackage{amsthm}
\usepackage{mdframed}

\newtheorem*{rep@theorem}{\rep@title}
\newcommand{\newreptheorem}[2]{%
\newenvironment{rep#1}[1]{%
 \def\rep@title{#2 \ref{##1}}%
 \begin{rep@theorem}}%
 {\end{rep@theorem}}}
\makeatother

\newreptheorem{theorem}{Theorem}

\newtheorem{question}{Question}
 
\begin{document}

\title{Quadratic worst-case message complexity for State Machine Replication  in the partial synchrony model}

% you can include author information in the source, but `anonymous' option will hide it

\author{Andrew Lewis-Pye}
%\email{a.lewis7@lse.ac.uk}
\affiliation{%
  \institution{\department{Department of Mathematics} \institution{London School of Economics} \city{London} \country{UK}}
}
\email{a.lewis7@lse.ac.uk}

\begin{abstract} 

We consider the message complexity of State Machine Replication protocols dealing with Byzantine failures in the partial synchrony model. A result of Dolev and Reischuk gives a quadratic lower bound for the message complexity, but it was unknown whether this lower bound is tight, with the most efficient known protocols giving worst-case message complexity $O(n^3)$. We describe a protocol which  meets  Dolev and Reischuk's quadratic lower bound, while also satisfying other desirable properties. To specify these properties, suppose that we have $n$ replicas, $f$ of which display Byzantine faults (with $n\geq 3f+1$). Suppose that $\Delta$ is an upper bound on message delay, i.e.\  if a message is sent at time $t$, then it is received by time $ \text{max} \{ t, GST \} +\Delta $. We describe a deterministic protocol that simultaneously achieves $O(n^2)$ worst-case message complexity, optimistic responsiveness, $O(f\Delta )$ time to first confirmation after $GST$ and $O(n)$ mean message complexity.

%In recent years there has been significant renewed interest in establishing State Machine Replication (SMR) protocols that are both robust and scalable. To establish robustness, a standard setting is to require protocols to tolerate Byzantine failures and to function in the partially synchronous setting. In terms of scalability, an important metric is the message complexity. A result of Dolev and Reischuk gives a quadratic lower bound for the message complexity, but it was unknown whether this lower bound is tight, with the most efficient known protocols giving message complexity $O(n^3)$. 
%
%In this paper we describe a protocol which meets  Dolev and Reischuk's quadratic lower bound while also satisfying other desirable properties. In particular, suppose we have $n$ replicas, $f$ of which display Byzantine faults (with $n\geq 3f+1$). Suppose that $\Delta$ is an upper bound on message delay, i.e.\  if a message is sent at time $t$, then it is received by time $ \text{max} \{ t, GST \} +\Delta $. We describe a deterministic protocol that simultaneously achieves $O(n^2)$ worst-case message complexity (after GST), optimistic responsiveness, $O(f\Delta )$ time to first confirmation after $GST$ and $O(n)$ mean message complexity.
%\begin{itemize} 
%\item  $O(n^2)$ worst-case authenticator complexity (after GST).
%\item Optimistic responsiveness;
%\item  $O(f\Delta )$ time to first confirmation after $GST$;
%\item $O(n)$ mean authenticator complexity.
%\end{itemize} 

  \end{abstract}

% keywords, ACM classification and conference information can be omitted for submission

\maketitle

\section{Introduction}

Protocols for State Machine Replication (SMR) \cite{lamport1978time,schneider1990implementing} are used to ensure that a distributed network of `replicas' (e.g.\ servers) can agree on an order in which to implement client-initiated service requests. Traditionally, such protocols have been studied in the `permissioned' setting, which means that the protocol  is carried out by a fixed set of replicas, all of whom are known to each other from the start of the protocol execution. Together with the related task of reaching Byzantine Agreement \cite{pease1980reaching,lamport1982byzantine}, SMR protocols are a central topic in distributed computing. For an up-to-date account of the latest advances in this area, see \cite{cohen2021byzantine}.

Of direct relevance to the study of SMR protocols is the fact that recent years have seen a surge of interest in `blockchain' technologies, such as Bitcoin \cite{nakamoto2008bitcoin}.  The key factor differentiating Bitcoin from traditional SMR protocols is that it functions in a `permissionless' setting, which means that the protocol operates over an unknown network of replicas that anybody is free to join. Despite this distinction,  the surge of investment in blockchain technology has led to significant renewed interest in permissioned SMR protocols. On the one hand, this is because many of the envisaged blockchain applications actually pertain to fundamentally permissioned settings  -- a prominent example is the use of SMR protocols to facilitate international transfers between members of banking consortia.  On the other hand, there are also now several established techniques (e.g. \cite{chen2016algorand}) for implementing permissioned protocols in permissionless settings. The latest implementations of some crytocurrencies, such as Ethereum \cite{buterin2018ethereum}, use permissioned SMR protocols \cite{buterin2017casper} that are tailored in a simple way to instantiate a form of permissionless entry.

These new applications bring with them a focus on the need for permissioned SMR protocols that are both robust and scalable. Robustness means being able to deal with arbitrary (especially malicious) behaviour from a non-trivial proportion of replicas. In the parlance of the distributed computing literature, protocols that can deal with behaviour of this kind are called `Byzantine Fault Tolerant' (BFT).  Practical solutions may also need to tolerate unbounded periods of bad network connectivity and denial-of-service attacks. A standard way to model this is to work in the `partial synchrony' communication model of \cite{DLS88}, which means that protocols are required to  guarantee consistency at all times (i.e. that non-faulty replicas never disagree on the ordering of client requests), but are only required to make progress on processing new client requests during time intervals when message delivery is reliable, i.e.\ during intervals when messages are guaranteed to be delivered within some known time bound $\Delta$.

\vspace{0,2cm} 
\noindent \textbf{The issue of scalability.} PBFT \cite{castro1999practical} is well-known as the first practical BFT SMR protocol for the partial synchrony model. When PBFT was first conceived, however, typical real-world SMR implementations would involve a number of replicas $n$ in the single digits.  Today, implementations are envisaged that use hundreds or even thousands of replicas. It therefore becomes vital to consider protocols that scale well with the number of replicas, and message complexity is an important metric in this regard.\footnote{Message complexity will be defined formally in Section \ref{model}. Roughly, it is the worst-case number of messages that need to be sent to confirm requests once message delivery is reliable.} A result of Dolev and Reishchuck \cite{dolev1985bounds} establishes a quadratic lower bound on the message complexity when the given bound $f$ on the number of faulty replicas is $\Theta (n) $ (as we assume here).  Many BFT SMR protocols have been developed since PBFT (e.g. Zyzzyva \cite{kotla2010zyzzyva}, Prime \cite{amir2010prime}, SBFT \cite{gueta2019sbft}, Upright \cite{clement2009upright}, 700BFT \cite{guerraoui2010next}, BFT- SMaRt \cite{bessani2014state}, Tendermint \cite{buchman2016tendermint,buchman2018latest}, Hotstuff \cite{yin2018hotstuff}, LibraBFT \cite{baudet2019state}, Casper \cite{buterin2017casper}), but none achieve worst-case message complexity better than $O(n^3)$. 

The aim of this paper is to describe a BFT SMR protocol for the partial synchrony model that achieves worst-case message complexity $O(n^2)$, and so which is the first such protocol to meet the quadratic lower bound of Dolev and Reishchuk. The protocol we describe also satisfies other desirable properties: Optimistic responsiveness, $O(f\Delta )$ time to first confirmation after $GST$ and $O(n)$ mean message complexity. These terms will be explained in Sections \ref{model} and \ref{rw}.  In the context of the partial synchrony model,  a protocol is  said to have \emph{optimal resiliency} if it can handle any number $f$ of Byzantine faults, so long as  $n\geq 3f+1$ (this being optimal by a result of Dwork, Lynch and Stockmeyer \cite{DLS88}). 

\begin{theorem} \label{t1} 
Consider the partial synchrony model. There exists a BFT SMR protocol with optimal resiliency  and  worst-case message complexity $O(n^2)$. This protocol can simultaneously be made to satisfy optimistic responsiveness, to give $O(f\Delta )$ time to first confirmation after $GST$ and $O(n)$ mean message complexity.
\end{theorem} 

Theorem \ref{t1} above is stated in terms of \emph{message} complexity, i.e.\ the number of messages sent. A more fine-grained analysis is achieved by thinking in terms of \emph{communication} complexity, i.e.\ the number of bits  sent to confirm requests. 
The variable $\kappa$ in the theorems below is a security parameter specifying the length of signatures and hashes. Theorem \ref{t2} gives an analogue of Theorem \ref{t1} in terms of communication complexity, but drops the requirement for $O(n)$ mean communication complexity. Theorem \ref{t3}  can be seen as a corollary of Theorem \ref{t2}. We will also give an independent proof, which is simpler than a direct proof of Theorem \ref{t2}. 

\begin{theorem} \label{t2} 
Consider the partial synchrony model. There exists a BFT SMR protocol with optimal resiliency  and  worst-case communication complexity $O(\kappa n^2)$. This protocol can simultaneously be made to satisfy optimistic responsiveness and to  give $O(f\Delta )$ time to first confirmation after $GST$. \end{theorem}

\begin{theorem} \label{t3} Consider the partial synchrony model. There exists a protocol for Byzantine Agreement with optimal resiliency  and  worst-case communication complexity $O(\kappa n^2)$. This protocol can be made to ensure that all correct replicas terminate within time  $O(f\Delta )$ after $GST$.
\end{theorem}

%Theorem \ref{t1} above refers to message complexity (i.e.\ the number of messages sent), rather than communication complexity (i.e.\ the number of bits sent). We will describe analogues of Theorem \ref{t1} for communication complexity in Section \ref{comp}. 

%As explained in Section \ref{model}, `message complexity' in Theorem \ref{t1} above, can also be replaced with `authenticator complexity' (as defined in that section). In fact, the same holds for  communication complexity (i.e.\ the number of bits sent), so long as one is prepared to sweep under the rug the fact that messages must contain various written \emph{epoch/view} numbers, which are theoretically of unbounded length, but will realistically be bounded by small values in any actual implementation. 

% Initial small replica sets. New applications want many more, so it's very important to scale to larger numbers. One significant metric is message complexity (which measures the number of messages required to agree on client requests during periods of reliable message delivery as a function of n, the number of replicas). Dolev Reischuk gives a quadratic lower bound. Since PBFT many examples (list quite a few), but these all have worst-case message complexity $\Omega(n^3)$. 
%
%The aim of this paper is to decribe a protocol that meets Dolev and Reischuks' quadratioc lower bound. We will also ensure other desirbale properties, such as optimisitc responsiveness, which will be described in Sections 1.2 and 1.3. 

%\begin{document}

 \subsection{Defining the model and metrics} \label{model} 
 
 \textbf{The replicas.} We consider a set of $n \geq  3f + 1$ replicas, indexed by $i\in  \{ 0,\dots, n-1 \}$. Certain replicas may be \emph{corrupted} by an \emph{adversary}, and may then display Byzantine failures. We allow that the adversary can choose which replicas to corrupt in a dynamic fashion over the course of the protocol execution, so long as they corrupt a total of at most $f$ many replicas. We will refer to the (at least $n-f$) replicas that are never corrupted as \emph{correct}. 
 Every pair of replicas has a bidirectional, reliable, and authenticated channel between them, i.e., each replica is able to send messages individually to each of the other replicas, these messages will eventually arrive, and the recipient can verify the sender’s identity. When we say that a replica  \emph{broadcasts} a message, this means that it simultaneously sends the message to all replicas.\footnote{It will be technically convenient to suppose that, when a replica broadcasts a message, it also `sends' the message to itself. This just means that, for the purpose of carrying out its instructions, the replica immediately regards the message as having been received. }

\vspace{0.1cm} 
\noindent \textbf{Cryptographic assumptions}. We assume the existence of a cryptographic signature scheme, a collision-resistant hash function $h$, a public key infrastructure (PKI) to validate signatures,  and a threshold signature scheme \cite{boneh2001short,cachin2005random,shoup2000practical}. As is standard, we assume that replicas are polynomial-time bounded, which means that replicas are given a security parameter $\kappa$ specifying the length of signatures and hashes, and may perform a number of basic operations during the protocol execution which is at most polynomial in $\kappa$. In particular, this means that the length of the protocol execution is bounded by a polynomial in $\kappa$.  As in other SMR protocols \cite{yin2018hotstuff}, the threshold signature scheme is used to create a  signature of size $\kappa$ combining signatures from $k$ of $n$ many replicas. In our case, $k$ will always take the value $n-f$. We suppose that hashes are unique during any protocol execution, i.e.\ for any two inputs $m$ and $m'$ given to the hash function, $h(m)=h(m')$ implies $m=m'$. We also restrict attention to protocol executions in which the adversary is unable to break the cryptographic schemes described above.

\vspace{0.1cm} 
\noindent \textbf{Communication models}. Three standard communication models are: 

\begin{enumerate} 
\item \emph{The synchronous model}.   All replicas begin the protocol execution simultaneously with access to a global clock. There is a known bound $\Delta$ on the time for message delivery. 

\item  \emph{The partial synchrony model}.  The known bound $\Delta$ on message delivery only applies after some unknown \emph{Global Synchronisation Time}\footnote{The model is really intended to model a scenario in which consistency must be maintained at all times, and new requests must be confirmed during synchronous intervals of \emph{sufficient length}. The use of the unknown Global Synchronisation Time is a standard technical convenience, which can be argued to be equivalent to these requirements. See \cite{DLS88} for further details.}  (GST), i.e.\ if a message is sent at time $t$ then it arrives by time $\text{max} \{ \text{GST}, t \} +\Delta$.
Replicas do not have access to a  global clock, but can accurately measure local time, i.e.\ each replica has an accurate stopwatch/timer.\footnote{This assumption is made for technical convenience, but it will be clear that the protocol we describe can easily be adjusted (without communication complexity cost) to handle scenarios in which there is a known upper bound on the ratio between the speed of the timers held by different replicas.} There is no requirement that replicas all start the protocol execution simultaneously, but all correct replicas are assumed to join prior to GST. 

%\item  \emph{The partial synchrony model}.  Replicas do not have access to a  global clock, and the known bound on message delivery only applies after some unknown \emph{Global Synchronisation Time}\footnote{The model is really intended to model a scenario in which consistency must be maintained at all times, and new requests must be confirmed during synchronous intervals of \emph{sufficient length}. The use of the unknown Global Synchronisation Time is a standard technical convenience, which can be argued to be equivalent to these requirements. See \cite{DLS88} for further details.}  (GST), i.e.\ if a message is sent at time $t$ then it arrives by time $\text{max} \{ \text{GST}, t \} +\Delta$.
%We assume that replicas are able to accurately measure local time, i.e.\ each replica has an accurate timer.\footnote{This assumption is made for technical convenience, but it will be clear that the protocol we describe can easily be adjusted to handle scenarios in which there is a known upper bound on the ratio between the speed of the timers held by different replicas.} There is no requirement that replicas all start the protocol execution simultaneously, but all correct replicas are assumed to join prior to GST.  

\item  \emph{The asynchronous model}. Each message can take any finite amount of time to be delivered. Replicas do not have access to timers. 

\end{enumerate} 

An important assumption in the description of the partial synchrony model above is that replicas do not have globally synchronised clocks. If replicas have access to globally synchronised clocks, then a simple approach to achieving the quadratic lower bound of Dolev and Reischuk in terms of message complexity (which works if one is willing to drop the requirement for optimistic responsiveness -- see Section \ref{rw} for a definition of the latter term) is to implement a version of Hotstuff in which each view is carried out at a previously determined time according to the global clock. More generally, the problem of achieving the quadratic lower bound is interesting so long as \emph{any} of the following three conditions hold: (1) Replicas  do not have globally synchronised clocks; (2) There is a requirement for optimistic responsiveness; (3) We are concerned with communication (rather than just message) complexity. We will consider dropping all of these assumptions in this paper. 

A seminal result of result of Fisher, Lynch and Paterson \cite{fischer1985impossibility} shows that, for the asynchronous model,  deterministic SMR protocols are not possible and that no SMR protocol (deterministic or otherwise) is able to achieve unbounded values for the worst-case complexity measures we consider here. We work in the partial synchrony model and it will be convenient to think of the adversary as controlling message delivery subject to the constraints outlined above.

\vspace{0.1cm} 
\noindent \textbf{The blockchain}. Rather than being concerned with the rate at which clients produce requests, we consider a protocol that produces chains of \emph{blocks}, each of which is thought of as containing a constant bounded number of requests. Replicas begin the protocol execution with agreement on a unique \emph{genesis} block. All blocks other than the genesis block must have a unique \emph{parent} block, whose hash is specified in the block. If $B'$ is the parent of $B$, then $B'$ is also called a \emph{predecessor} of $B$, and all predecessors of $B'$ are predecessors of $B$. All blocks must have the genesis block as a predecessor -- to avoid constant mention of replicas having to check for block `validity', it will be convenient in what follows to assume replicas simply ignore blocks that cannot be verified to satisfy these conditions. Two distinct blocks are called \emph{incompatible} if neither is a predecessor of the other. 

\vspace{0.1cm} 
\noindent \textbf{Consistency and liveness}. We consider a protocol that \emph{confirms} blocks. The requirement for \emph{consistency} is that no pair of incompatible blocks should ever be confirmed. 
The requirement for \emph{liveness} is that arbitrarily many blocks will be confirmed if the protocol is run for sufficiently long after GST -- for the protocols we describe here, we will be able to guarantee that each interval of length $O(f\Delta)$ after GST produces a new confirmed block.

%\vspace{0.1cm} 
%\noindent \textbf{The complexity measures}. A number of papers in the area (e.g.\ \cite{yin2018hotstuff}) consider a measure which is more fine-grained than message complexity, which is called \emph{authenticator} complexity, and which measures the number of signatures sent by correct replicas. Since we assume that all messages are signed, the authenticator complexity also upper bounds message complexity. In fact, in the blockchain building protocol we describe in Section \ref{protocol}, all messages sent would be of bounded length, were it not for the fact that they contain various written \emph{epoch/view} numbers which are theoretically of unbounded length. So, if one makes the assumption that the written epoch numbers are of bounded length (which is a realistic assumption in any real implementation), then all messages sent by correct replicas will be of bounded length. In truth, the latter assumption is no stronger than our existing assumption that the hash function is perfectly collision resistant, because if the written epoch numbers exceed the length of hash values, then there must exist two different epoch numbers with the same hash. 
%For the sake of clarity, we will make our definitions and conduct our analysis here in terms of authenticator complexity, but will come back to discuss these issues in Section \ref{comp}. 

\vspace{0.1cm} 
\noindent \textbf{The complexity measures}. 
In the protocols we describe here, messages will contain \emph{epoch/view} numbers. Since we assume that replicas are polynomial-time bounded, the written forms of epoch/view numbers are of length $O(\kappa)$.

\vspace{0.1cm} 
It follows from the result of Fisher, Lynch and Paterson, mentioned above, that SMR protocols in the partial synchrony model cannot guarantee the confirmation of requests prior to GST$+\Delta$.  For this reason, the worst-case complexity measure described in the following definition counts the number of messages sent by correct replicas after $\text{GST}+ \Delta$.  The maximum value in the definition is taken over all possible actions of the adversary and choices for GST (which we can suppose chosen by the adversary). 

\begin{definition} \label{wc} 
The worst-case message complexity is the maximum number of messages sent by correct replicas after $\text{GST}+\Delta$ and prior to the first block confirmation.  
\end{definition}

We also consider a version of Definition \ref{wc} which is modified in the obvious way to give the corresponding notion for communication complexity. 
Let  $B^{\ast}_i$ be the unique block (if it exists)  which is the $i$th to be confirmed after $\text{GST}+\Delta$. 
In the following definition, $x_i$ denotes the number of messages sent by correct replicas after the confirmation of $B_i^{\ast}$ and before the confirmation of $B_{i+1}^{\ast}$. The slightly circuitous wording of the definition stems from the fact that the adversary may prevent  $ (1/k) \sum_{i=1}^{k}  x_i$ from moving towards a limit value as $k$ increases, although it always remains bounded.

\begin{definition} \label{meandef} 
For $i\in \mathbb{N}_{>0}$, let $x_i$ be defined as above. We say the mean message complexity is $O(n)$ if there exists some constant $C$ and a function $k(n)$ such that, for all $n$ and all  $k>k(n)$,  $\frac{1}{k} \sum_{i=1}^{k}  x_i <Cn$ for all protocol executions in which at least $k$ blocks are confirmed after GST$+\Delta$. 
\end{definition}

 \subsection{Related work} \label{rw} 
 
 %  A recent paper by Momose and Ren \cite{momose2020optimal} establishes a result which is analogous to Theorem \ref{t1}, but for Byzantine Agreement in the synchronous model, i.e.\ they describe a protocol with optimal resiliency and quadratic communication complexity for the synchronous model. 
 In this paper, we are concerned with the partial synchrony model.
 PBFT was the first practical BFT SMR protocol for the partial synchrony model. As is now standard for BFT SMR protocols in the partial synchrony model, the  protocol implements a sequence of \emph{views}, with each view having a designated leader who is responsible for driving consensus on the next set of requests to be confirmed. Subsequent to PBFT a great number of BFT SMR protocols have been developed (e.g.\  Zyzzyva \cite{kotla2010zyzzyva}, Prime \cite{amir2010prime}, SBFT \cite{gueta2019sbft}, Upright \cite{clement2009upright}, 700BFT \cite{guerraoui2010next}, BFT- SMaRt \cite{bessani2014state}) which work by using the same paradigm of leaders, views and view changes, but none achieve worst-case message complexity better than $O(n^3)$. 

%  Zyzzyva \cite{kotla2010zyzzyva}, Prime \cite{amir2010prime}, SBFT \cite{gueta2019sbft}, Upright \cite{clement2009upright}, 700BFT \cite{guerraoui2010next}, BFT- SMaRt \cite{bessani2014state}, Tendermint \cite{buchman2016tendermint,buchman2018latest}, Hotstuff \cite{yin2018hotstuff}, LibraBFT \cite{baudet2019state}, Casper \cite{buterin2017casper}),
% 
%  PBFT was the first practical BFT SMR protocol for the partial synchrony model. As is now standard for BFT SMR protocols in the partial synchrony model, the  protocol implements a sequence of \emph{views}, with each view having a designated leader who is responsible for driving consensus on the next set of requests to be confirmed. Each view requires correct replicas to send $O(n^2)$ signatures, while \emph{view changes} (which are required because leaders may be faulty) have a worst-case signature footprint which is $O(n^3)$. The total number of signatures sent by correct replicas if $O(n)$  view-changes occur before a single consensus decision is reached is thus $O(n^4)$ in PBFT. 
 
For PBFT-like protocols, such as those listed above, the complexity bottleneck tends to be the part of the protocol dedicated to \emph{view changes} (which are required because leaders may be faulty). Recently, a number of elegant protocols, such as Tendermint \cite{buchman2016tendermint,buchman2018latest} and Casper \cite{buterin2017casper} have been described, which considerably simplify the leader replacement protocol.  In the case of Tendermint and Casper, however, this simplicity is achieved at the expense of \emph{optimistic responsiveness}.   Informally,  optimistic responsiveness \cite{yin2018hotstuff} requires that, for a fixed input $\Delta$, an unbounded number of blocks/requests can be confirmed in any time interval of fixed length after GST, so long as leaders are correct, and so long as messages are delivered and instructions carried out fast enough. So, when leaders are correct, the rate at which blocks are confirmed after GST depends on the actual speed of the system, and is not limited by high estimates for $\Delta$ (although liveness with Byzantine failures still depends on the bound $\Delta$ holding after GST). 

The significant achievement of Hotstuff \cite{yin2018hotstuff}, was to simultaneously achieve optimistic responsiveness, while keeping a linear bound on communication complexity \emph{within} each view. Liveness for Hotstuff, however, still requires implementing a protocol to ensure \emph{view synchronisation}, i.e.\ to ensure that all replicas eventually spend long enough within the same view with correct leader. In Hotstuff,  this task of view synchronisation is delegated to a separate \emph{Pacemaker} module. The implementation of the Pacemaker module is left unspecified, although it is pointed out that a practically infeasible method of using expontially increasing `timeouts' would suffice. 

The task of achieving a general and efficient implementation of the Pacemaker module is addressed in a couple of recent papers by Naor et.\ al. \cite{naor2019cogsworth} and Naor and Keidar  \cite{naor2020expected}. The first of these two papers describes a version of the Pacemaker module named Cogsworth, which, when faced with benign failures, achieves  $O(n)$ mean message complexity. In the face of Byzantine failures, however, Cogsworth requires $\Omega (n^2)$ messages per view change, which means that a sequence of $\Omega (n)$ many view changes before finding a correct leader requires sending $\Omega (n^3)$ many messages. The module described in \cite{naor2020expected} then improves on the ability to deal with Byzantine failures, achieving $O(n)$ mean message complexity in the face of Byzantine failures, and $O(\Delta)$ expected time to first block confirmation after GST (when implemented with a protocol such as Hotstuff). Since the protocol utilises a random process of leader selection it has unbounded worst-case message complexity, but if modified in the obvious way (using a deterministic sequence of rotating leaders) gives a deterministic protocol with worst-case message complexity which is $\Omega(n^3)$, with time to first confirmation after GST which is $O(f^2\Delta)$. 

The study of protocols for Byzantine Agreement was introduced in \cite{pease1980reaching,lamport1982byzantine}, where it was shown that the problem can be solved in the synchronous model iff $f<n/2$ when a PKI is provided. A recent paper by Momose and Ren \cite{momose2020optimal} describes a protocol with optimal resiliency and  quadratic communication complexity for the synchronous model, thereby meeting the lower bound of Dolev and Reischuk \cite{dolev1985bounds} for that model.  For the partial synchrony model, Dwork, Lynch and Stockmeyer \cite{DLS88} showed that Byzantine Agreement can be solved iff $f<n/3$. The state-of-the-art in terms of message/communication complexity for the partial synchrony model is achieved by using SMR protocols to solve Byzantine Agreement. The most efficient previously known protocols have worst case message  complexity $O(n^3)$ -- this can be achieved, for example, by implementing Hotstuff with any Pacemaker module that has message complexity $O(n^2)$ per view-change. $O(n)$ \emph{expected} message complexity can be achieved by combining Hotstuff with the Pacemaker module of Naor and Keidar \cite{naor2020expected}.

We refer the reader to \cite{cohen2021byzantine} for a  survey which also goes into detail re the state-of-the-art for probabilistic protocols in the asynchronous setting, and well as further results on protocols for achieving Byzantine Agreement.

\section{The protocol for Theorem \ref{t1}} \label{protocol} 
As is standard, we describe a protocol for which the instructions are divided into a sequence of \emph{views}. Within each view, the instructions will be essentially the same as for Hotstuff -- where the protocol differs from Hotstuff is in the mechanisms used to achieve view synchronisation. We assume the reader is familiar with Hotstuff, but  the role of Hotstuff in our protocol can really be entirely blackboxed. All one needs to understand about Hotstuff to follow the protocol we describe here, is that it achieves the following: 
\begin{itemize} 
\item Consistency: Incompatible blocks will never be confirmed so long as each individual replica never executes view instructions out of order, i.e. so long as no correct replica executes instructions for a view $v$ and then instructions for a view $v'<v$. 
\item Liveness: Block confirmation will occur, so long as all replicas are eventually in the same view  for sufficiently long and that view has a correct leader. 

\end{itemize}  

Since Theorem \ref{t1} refers to message rather than communication complexity, we don't really need to use a threshold signature scheme for the protocol of this section. We do so to facilitate our discussion in Section \ref{commy}, which describes how to modify the protocol to deal with communication complexity. 

\subsection{The high level idea.} The basic plan is that we will  spend $O(n^2)$ messages to synchronise,  but will do so only every $f+1$ views. Since at least one leader amongst those $f+1$ views will be correct and will produce a new confirmed block (if the view is initiated after GST), and since the  complexity of each view is $O(n)$, this results in $O(n^2)$ messages being required before a correct leader produces a new confirmed block. 

\subsection{Epochs.} As well as being divided into views, the instructions are divided into \emph{epochs}. Each epoch consists of $f+1$ views that are specific to that epoch.  Each epoch therefore has $f+1$ different leaders  (one for each view within the epoch), at least one of which must be correct. Leaders are rotated, so that the leaders for epoch $i$ are replicas $i\text{ mod }n, \dots, i+f \text{ mod }n$. We consider three different kinds of \emph{certificate}, each of which combines a set of $n-f$ signatures  into a single threshold signature: 
\begin{itemize} 
\item QCs (quorum certificates) correspond to a specific round of voting within a view. 
\item VCs (view certificates) combine votes to begin a new view. 
\item ECs (epoch certificates) combine votes to begin a new epoch.   
\end{itemize} 

Just as for Hotstuff, each view has three stages of voting. We order QCs by epoch number, then by  view number within the epoch, and then by stage number within the view. Each QC will \emph{correspond} to a specific block. All votes and blocks are assumed to be signed by the sender, and to contain the corresponding epoch, view and stage number.

Before describing how to coordinate the views within an epoch, we'll first describe how to implement epoch changes. We start that way because understanding epoch changes does not require understanding the rest of the protocol, and because the instructions for epoch changes are very simple. This simplicity stems from the fact that we can afford to spend $O(n^2)$ messages for each epoch change.   Roughly, when any replica \emph{wishes to enter epoch} $e$, it sends an ``$\mathtt{epoch}\  e$'' message to each of the $f+1$ leaders of epoch $e$. When a correct leader receives  $n-f$ $\mathtt{epoch}\  e$ messages, it combines those messages into a message with a single threshold signature, which we call an EC for epoch $e$. The leader then broadcasts that EC, signalling that replicas should begin the epoch. When any correct replica first sees an EC for epoch $e$, they broadcast that EC and begin the epoch. The precise instructions are below. 

\vspace{0.5cm} 
\begin{mdframed}

 \paragraph{The instructions for epoch changes.}   \label{4}
\begin{itemize} 
\item  When any replica \emph{wishes to enter epoch} $e$, it sends an ``$\mathtt{epoch}\  e$'' message to each of the $f+1$ leaders of epoch $e$ (potentially including itself). The conditions under which a replica \emph{wishes to enter} an epoch will be described  subsequently in the instructions for view changes.
\item If a replica is one of the $f+1$ leaders for epoch $e$, they enter epoch $e$ if they are presently in a lower epoch\footnote{If at any point any replica finds that the requirements to enter two different views or epochs are simultaneously fulfilled, then the replica enters the higher of the two.} and either: 
\begin{enumerate} 
\item[(a)] They receive $n-f$ $\mathtt{epoch}\  e$ messages (from different replicas), or; 
\item[(b)] They receive an EC for epoch $e$.
\end{enumerate} 
 In the case that (a) holds, the leader combines the $n-f$ $\mathtt{epoch}\  e$ messages into a message with a single threshold signature, which we call an EC for epoch $e$, and the leader then broadcasts that EC. In the case that (b) holds, the leader simply broadcasts the EC. 
\item Each replica which is not a leader for epoch $e$ enters the epoch if they are presently in a lower epoch and if they see an EC for epoch $e$. Upon entering epoch  $e$, they broadcast that EC for epoch $e$ to all replicas.
\item A replica leaves an epoch when it enters a higher epoch. 
\end{itemize} 

\end{mdframed} 

\vspace{0.5cm} 
The instructions described above  ensure: 
\begin{enumerate} 
\item[$(\dagger_0)$]  If a correct replica $i$ enters epoch $e$ at time $t$, then all correct replicas will have entered an epoch $e'\geq e$ by time $\text{max} \{ t, GST \} +\Delta$, since they will see the EC sent by $i$ by that time. 
\item[$(\dagger_2)$] Each epoch change has message complexity $O(n^2)$, since each replica sends $O(f)$ messages upon wishing to enter epoch $e$ and $O(n)$ messages upon entering epoch $e$. 

\end{enumerate}

\subsection{Proceeding through the views within an epoch.} A simple approach to coordinating the views within an epoch would be to require correct replicas to restart their timer upon entering the epoch, and then to begin each view at a designated time according to their timer (so that timers are reset once per epoch). So long as all correct replicas begin the epoch within time $\Delta$ of each other, this approach would also (sufficiently) synchronise correct replicas for each view within the epoch.  The problem with this simple approach is that it would negate the possibility for optimistic responsiveness -- for the protocol to be optimistically responsive, we need correct leaders to be able to produce confirmed blocks at a rate that depends on how fast messages are delivered, rather than at a maximum speed which is determined by $\Delta$. To achieve optimistic responsiveness, we therefore also allow replicas to begin a view $v+1$ earlier than the designated time, if a confirmed block corresponding to view $v$ is produced. 

 Roughly, the instructions are as follows. When each replica starts epoch $e$, they  reset their personal timer to 0 and begin by ``wishing to enter view 0'' (of epoch $e$). Each replica will wish to enter a view $v>0$  when either their timer reaches the designated time for that view, or else they see a block corresponding to view $v-1$ confirmed. When they wish to enter a certain view $v$, they send a ``$\mathtt{view}\ v$'' message to the leader of view $v$. When the leader sees $n-f$ $\mathtt{view}\ v$ messages, they form them into a VC for that view, and broadcast that VC to all replicas. Replicas will begin view $v$ when they see a VC for the view (unless they are already in a higher view).  The precise instructions, which are described below, refer to the  `message state' of a replica: At any point in the protocol execution, the message state of a replica is the set of all messages received by that replica so far (including messages they have `sent' to themselves). 

\vspace{0.5cm} 

\begin{mdframed}

 \paragraph{The instructions for view changes.}   \label{ecs} 
 
\begin{itemize} 
\item When each replica starts epoch $e$, they  reset their personal timer to 0 and begin by  \emph{wishing to enter view} 0.

\item A replica \emph{wishes to enter view} $v>0$ of epoch $e$ at the first point at which it is in a lower view (for this epoch) and either: 
\begin{itemize} 
\item[(a)] It sees a block corresponding to view $v-1$ confirmed, or; 
\item[(b)] Its personal timer reaches $12v \Delta$ (the idea to have in mind here is that each view will  last for at most $12\Delta$ after $GST$ when leaders are correct). We'll refer to  $12v \Delta$ as the \emph{trigger time} for view $v$. 
\end{itemize} 

\item  When any replica \emph{wishes to enter view} $v$, it sends a $\mathtt{view}\  v$ message to the leader of view  $v$.  It attaches to that message its present message state (which includes the highest QC it has seen). 
%highest QC it has seen, together with the hash of the block to which that QC corresponds. (for use later)
\item The leader of view $v$ of epoch $e$ enters the view if they are in epoch $e$ and both: 
\begin{enumerate} 
\item[(a)] $i=0$ or they are presently in a lower view, and;
\item[(b)]  They receive $n-f$ $\mathtt{view}\  v$ messages (from different replicas).
\end{enumerate} 
 When this occurs, the leader combines the $n-f$ $\mathtt{view}\  v$ messages into a message with a single threshold signature, which we call a VC for view  $v$. The leader considers the highest QC it has seen, which corresponds to a block $B'$ (say). The leader then broadcasts a new proposed block $B$ with parent  $B'$, together with the newly produced VC, and attaches to that message its present message state. 
 \item Each replica which is not the leader of view $v$ of epoch $e$ enters the view if they are in epoch $e$ and both: 
 \begin{enumerate} 
 \item[(a)] $v=0$ or they are presently in a lower view, and; 
 \item[(b)] they see a VC for view $v$. 
 \end{enumerate} 

\item Each replica in epoch $e$ wishes to enter epoch $e+1$ the first time that either it sees a confirmed block for view $f$ of epoch  $e$, or else its personal timer reaches $12(f+1)\Delta$. Later, we'll also refer to $12(f+1)\Delta$ as a \emph{trigger time}. 

\end{itemize} 

\end{mdframed}

\vspace{0.5cm}
The instructions above suffice to describe how views and epochs are coordinated. It only remains to describe the instructions within each view, which are similar to Hotstuff. 

\subsection{The instructions within each view.} In addition to the instructions for view and epoch changes above, the remaining instructions are as follows. These instructions are similar to Hotstuff and, in particular, make use of the same `lock' functionality (we refer the reader to \cite{buchman2016tendermint,yin2018hotstuff} for an explanation of this functionality).    All replicas begin in epoch 1 with their `lock' set to be the genesis block, which corresponds to view 0 of epoch 0.  The `highest QC' \emph{below} a block $B$ is the highest QC that corresponds to any block (including $B$) in the initial segment of that chain. 
The instructions for each view are then as follows. 
\vspace{0.5cm} 

\begin{mdframed} 

\paragraph{The instructions while in view $v\geq 0$} At any time while in view $v$, until such a point as it wishes to enter a greater view or epoch (or actually enters a greater view or epoch), each replica carries out the following instructions. 
\begin{itemize} 
\item  \textbf{Stage 1 Voting.} Each replica considers the first block $B$ they receive for view $v$ produced by the leader (this arrives with the VC for the view). If it extends their lock, or else if the highest QC below the block is as high or higher than they have previously seen (if it's higher they release their lock), they send a stage 1 vote for $B$ to the leader. The stage 1 vote for $B$ is a signed message including the stage, view and epoch number and a hash of the block. 
\item If the leader receives $n-f$ stage 1 votes for the block they proposed, then they combine these votes into a stage 1 QC for the block, and broadcast this QC. 
\item  \textbf{Stage 2 Voting.}  The first time a replica sees a stage 1 QC for the block $B$ (specified above) they send a stage 2 vote for $B$ to the leader. 
\item If the leader receives $n-f$ stage 2 votes for the block they proposed, then they combine these votes into a stage 2 QC for the block, and broadcast this QC. 
\item   \textbf{Stage 3 Voting.} The first time a replica sees a stage 2 QC for the block $B$ (as specified above) they set $B$ as their lock and send a stage 3 vote for $B$ to the leader.
\item If the leader receives $n-f$ stage 3 votes for the block they proposed, then they combine these votes into a stage 3 QC for the block (which marks the block as confirmed), and broadcast this QC. 
\end{itemize} 

\end{mdframed}

\section{Consistency, liveness and complexity} \label{proofsection} 

%\subsection{Liveness and consistency} 

First of all, we deal with consistency, which works exactly as for Hotstuff. For the sake of making our paper as self-contained as possible, we include a proof here. 

\begin{lemma} \label{consistency} 
The protocol described in Section \ref{protocol} satisfies consistency. 
\end{lemma} 

\begin{proof}  Suppose, towards a contradiction, that two incompatible blocks $B$ and $B'$ are both confirmed. During the description of the protocol in Section \ref{protocol}, we numbered views within each epoch, so that each epoch has views numbered from 0 to $f$. Just  for the duration of this proof, however, it is convenient to suppose that views all have distinct numbers, so that epoch 1 contains views $0,\dots,f$, then epoch 2 contains views $f+1,\dots, 2f+1$, and so on. Note also, that blocks and votes all contain their epoch, view and stage number. So, we can talk of blocks and votes as corresponding to specific views (and stages of voting). 
Since any two sets of $n-f$ replicas must have a correct replica in the intersection, and since each correct replica only sends at most one stage 3 vote in each view, it immediately follows that $B$ and $B'$ must correspond to different views.  Without loss of generality, suppose that $B$ corresponds to view $v$  and that $B'$ corresponds to a greater view $v'>v$. Let $v''$ be the least view $>v$ such that a block $B''$ corresponding to $v''$ is incompatible with $B$ and receives a stage 1 QC. The fact that $B$ is confirmed means that at least $n-2f$ correct replicas must set $B$ as their lock while in view $v$. None of those $n-2f$ correct replicas which set $B$ as their lock during view $v$ can produce a stage 1 vote for $B''$ during view $v''$. Since any two sets of $n-2f$ correct replicas contain a correct replica in the intersection, this means that $B''$ cannot receive a stage 1 QC, and gives the required contradiction.
\end{proof} 

Next, we prove liveness. As we do so, however, it is convenient to consider also the message complexity: 

\begin{lemma} \label{liveness} 
The protocol described in Section \ref{protocol} satisfies liveness, $O(n^2)$ worst-case message complexity and $O(f\Delta )$ time to first confirmation after $GST$. 
\end{lemma} 
\begin{proof}  Let $e_0$ be the greatest epoch that any correct replica is in at any time $<\text{GST}$, let $t_1$  be the first time (if there exists such) at which any correct replica enters epoch $e_0+1$, and let
$v_1$ be the least view of epoch $e_0+1$ which has a correct leader. Our basic aim is to prove that view $v_1$ of epoch $e_0+1$ will produce a confirmed block. To do so, we first want to produce a bound on $t_1$. 

\vspace{0.2cm} 
\noindent \textbf{Bounding} $t_1$. Note first that, by $(\dagger_0)$ (of Section \ref{protocol}), all correct replicas will see an EC for epoch $e_0$ by  $GST +\Delta$, and will be in an epoch $\geq e_0$ by this time. Generally, for any correct  replica to enter an epoch $e>1$, at least $n-2f$ correct replicas have to be in epoch $e-1$ and wish to enter epoch $e$. This means that the first epoch $e>e_0$ that any correct replica will enter at any time $\geq GST$ will be $e_0+1$. In fact, we can argue that $t_1\leq GST +2 \Delta + 12(f+1)\Delta $. To see this, suppose towards a contradiction that no correct replica enters epoch $e_0+1$ (and therefore no higher epoch either) by time  $GST +2 \Delta + 12(f+1)\Delta $. Then all correct replicas will wish to enter epoch $e_0+1$ by time $GST+ \Delta +   12(f+1)\Delta$, and will send an $\mathtt{epoch}\ e_0+1$ message to all leaders of epoch $e_0+1$. A correct leader for epoch $e_0+1$ will then receive $n-f$ $\mathtt{epoch}\ e_0+1$ messages by time $GST+ 2\Delta +   12(f+1)\Delta$, and will enter epoch $e_0+1$ at this time.

\vspace{0.2cm} 
\noindent \textbf{Proving that view $v_1$ produces a confirmed block}.   Let $t^{\ast}$ be the first time (if it exists) at which a block corresponding to view $v_1$ of  $e_0+1$ is confirmed, and let $t_2 = \text{min} \{ t^{\ast}, t_1+12(v_1+1)\Delta \}$. Note that, for a correct replica to enter a view or epoch ($>1$), at least $n-2f$ correct replicas must first wish to enter the respective view or epoch. Correct replicas will only wish to enter a given view or epoch when either they see a block corresponding to the previous view\footnote{Here, it is to be understood that the view previous to view 0 of epoch $e$ is view $f$ of epoch $e-1$.} confirmed, or else their timer reaches the relevant `trigger time'  (see the instructions for view changes in Section \ref{protocol}). This means that:

\begin{itemize} 
\item[($\diamond_0)$] No correct replica can enter an epoch $>e_0+1$ prior to $t_2$. 
\item[($\diamond_1)$] No correct replica can enter any view $v>v_1$ of epoch $e_0+1$ prior to $t_2$. 
\end{itemize} 

We claim that:

\begin{enumerate} 
\item[$(\diamond_2)$]  The number of messages sent by correct replicas (combined) in the time interval $(\text{GST}+\Delta, t_2]$ is $O(n^2)$. 
\item[$(\diamond_3)$]  A block corresponding to view $v_1$ will be confirmed by time $t_2$.
\end{enumerate} 

From $(\diamond_2)$ and $(\diamond_3)$ it follows that $O(n^2)$ messages are sent by all correct replicas combined after $GST+\Delta$ and before the first block proposed by a correct leader is confirmed. It also follows that the  time to first confirmation after $GST$ is  $O(f\Delta )$,  since  $t_1\leq GST +2 \Delta + 12(f+1)\Delta $ and $t_2 \leq  t_1+12(v_1+1)\Delta$.

Statement $(\diamond_2)$ follows straight from the definition of the protocol -- correct replicas only carry out at most one change of epoch (from $e_0$ to $e_0+1$) in the time interval $(\text{GST}+\Delta, t_2]$, and each of the $O(n)$ many views within an epoch has the correct replicas combined sending $O(n)$ many messages. 

Next we show $(\diamond_3)$. This now follows essentially just as in Hotstuff. Suppose, towards a contradiction, that no block corresponding to view $v_1$  is confirmed by time $t_1 +12(v_1+1)\Delta$. By $(\diamond_0)$, it follows that all correct replicas will be in epoch $e_0+1$ by time $t_1+\Delta$. By $(\diamond_0)$ and $(\diamond_1)$, it follows that all correct replicas will be in view $v_1$ of epoch $e_0+1$ by time $t_1+3\Delta +12v_1\Delta$ (and within time $\Delta$ of each other). Before broadcasting the VC for view $v_1$, the leader receives $\mathtt{view}\ v_1$ messages from at least $n-2f$ correct replicas, and each of these messages includes the highest QC that the replica has seen at the point that the message is sent. When correct replicas wish to enter view $v_1$ and send a $\mathtt{view}\ v_1$ message, they cease executing instructions for previous views. Now consider the highest lock of any correct replica $i$  before receiving the block $B$ from the leader of view $v_1$. This lock was set because $i$ saw a stage 2 QC for a block $B'$. This means that at least $n-2f$ correct replicas must have produced stage 2 votes for $B'$ (prior to any point at which they wished to enter view $v_1$ and stopped executing instructions for views $<v_1$), and so must have seen a stage 1 QC for $B'$ by that time. Since any two sets of $n-2f$ correct replicas have a correct replica in the intersection, it follows that one of the $\mathtt{view}\ v_1$ messages seen by the leader of $v_1$ must be from a correct replica who had already seen the stage 1 QC for $B'$. 
This means that all correct replicas will then produce stage 1,2 and 3 votes for the block proposed by the  leader of view  $v_1$, and this block will be confirmed before time  $t_1 +12(i+1)\Delta$, giving the required contradiction. 
\end{proof} 

\vspace{0.2cm} The fact that the protocol is optimistically responsive is clear from the protocol description, so we are left to consider mean message complexity. 

\begin{lemma} 
The protocol described in Section \ref{protocol} has $O(n)$ mean message complexity.
\end{lemma} 
\begin{proof} Let $e_0$ be as defined in the proof of Lemma \ref{liveness}, i.e.\ let $e_0$ be the greatest epoch that any correct replica is in at any time $<\text{GST}$.  The proof of Lemma \ref{liveness} actually suffices to establish that all correct leaders for views in epochs $>e_0$ will produce confirmed blocks (just redefine $t_1$ in the proof to be the first time at which any correct replica enters the relevant epoch and omit the parts of the proof that are now not relevant). The lemma therefore follows, since epochs $>e_0$ will confirm $\Omega(n)$ many blocks on average and will each require correct replicas to send $O(n^2)$ many signatures. 
\end{proof}

\section{Communication complexity} \label{commy} 

\subsection{The issue in dealing with communication complexity.} All messages sent by replicas in the protocol of Section \ref{protocol} would be of length $O(\kappa)$, were it not for the fact that the instructions required some messages to have attached the message state of the replica.\footnote{Recall from Section \ref{model} that the written versions of epoch numbers are $O(\kappa)$.} When replicas send a $\mathtt{view}\ v$ message, for example, they are required to attach their message state. In fact, the instructions differ from the original Hotstuff protocol in that regard -- in the Hotstuff protocol, replicas are only required to attach the highest QC they have seen to the message (which is called a NEW-VIEW message in that paper). The question then becomes, if one follows the instructions for Hotstuff, is the information received by the leader actually enough to produce the next block? If the leader is to produce a block which does not repeat requests from earlier in the blockchain, then it isn't enough just to see the highest QC, they also need to know all blocks which are predecessors of the block to which that QC corresponds. One could allow that requests be repeated in the blockchain, and dictate that only the first instance of a request should count (with clients numbering different instances of requests to avoid confusion), but this does not negate the fact that any replica actually needs to know all blocks in the chain to carry out the requests in order. So, there is a subtle difference between the information that is required to build the blockchain and the information that is actually required to execute the requests, in that case. Generally, one has the issue that (even after GST), faulty leaders may confirm $f$ consecutive blocks, with up to $f$ many of the correct replicas not being aware of those blocks. If the protocol is to remain live, then those correct replicas certainly have to be made aware of the confirmed blocks \emph{at some point}. If one has leaders routinely broadcasting $f$ many predecessors each time they broadcast a new block, then each view will have communication complexity $\Theta(\kappa n^2)$. If one has previous blocks only sent upon demand by the relevant replica, then it seems difficult to avoid faulty replicas making demands that cause communication complexity  $\Theta(\kappa n^2)$ per view. 

\subsection{Byzantine Agreement.} To deal with this issue, we first observe that such considerations become unproblematic in the case that we only wish to satisfy Byzantine Agreement \cite{lamport1982byzantine}. Recall that, in the task of reaching Byzantine Agreement, each replica  is given an input (which we assume to have length $O(\kappa)$), that all correct replicas are required to terminate with the same output, and that if all correct replicas are given the same input $u$, then $u$ must be their common output.

\vspace{0.2cm} 
To prove Theorem \ref{t3}, we make the following changes to the protocol of Section \ref{protocol}:
\begin{itemize} 
\item The protocol begins with an extra instruction that has each correct replica broadcast a signed message with its input. 
\item Blocks no longer contains requests. Any block that has the genesis block as parent must contain signed inputs from $n-f$ processors. The \emph{decision value} corresponding to this set of signed inputs is the most common signed value, with ties broken by least hash. If a correct replica is unable to propose a new block when instructed to do so, it waits for time $\Delta$ and then omits the instruction if still unable to proceed. 
\item Blocks which do not have the genesis block as parent contain only their epoch and view number and a hash of their parent block, together with a decision value, which must be the same as the decision value for their parent block (otherwise the block is not valid and will be ignored by correct replicas). 
\item When a replica sends a $\mathtt{view}\ v$ message, it now attaches only the highest QC it has seen, together with the block corresponding to that QC. 
\item When a leader proposes a new block $B$ with parent block $B'$ other than the genesis block, it now broadcasts only $B,B'$ and the QC for $B'$, together with the relevant VC. Correct replicas will ignore the new block unless it arrives with $B'$, a QC for $B'$, and unless the decision and hash values match. If $B$ has the genesis block as parent, then the leader broadcasts only the signed block $B$, together with the relevant VC. 
\item When a correct replica sees a confirmed block with decision value $u$, it broadcasts the block together with the stage 3 QC, before terminating with output $u$.  
\end{itemize} 

\noindent \textbf{Verification}. Since correct replicas will ignore a new proposed block unless either it has parent the genesis block or else they are able to verify that it has the same decision value as its parent, it follows that no block can receive a QC unless it records the same decision value as its parent or has the genesis block as parent. For any block $B$ with a QC, it follows inductively that all predecessors (other than the genesis block)  also receive QCs, and so all record the same decision value. 
The proof of Lemma \ref{consistency}, which was oblivious to the content of blocks beyond the fact that they should satisfy certain structural properties to be valid,  suffices to show that incompatible blocks cannot be confirmed, which means that all confirmed blocks must have the same decision value. Since all signed inputs from correct replicas will be delivered by GST$+\Delta$, the proof of Lemma  \ref{liveness} suffices to show that a confirmed block will be produced and all correct replicas will terminate within time $O(f\Delta )$ after $GST$. Since $O(n^2)$ messages will be sent after GST$+\Delta$ and prior to the first point at which a correct replica sees a confirmed block, and since all messages are of length $O(\kappa)$, the worst-case communication complexity is $O(\kappa n^2)$.  If all correct replicas receive the same input $u$,  then the decision value of any (valid) block with the genesis block as parent must be $u$. It follows that any block receiving a QC must have decision value $u$.

\subsection{Proving Theorem \ref{t2}.}  In the following proof, we assume either that all requests are of length $O(\kappa)$, or else that the length of requests is discounted when counting communication complexity. As described in Section \ref{model}, we also suppose that blocks contain a constant bounded number of requests. The basic idea is that we carry out repeated instances of Byzantine Agreement, using a protocol like that described in the proof of Theorem \ref{t3}. 
We make the following changes to the protocol of Section \ref{protocol}. 

\vspace{0.2cm} 
\noindent (1) We consider \emph{super-epochs}: 
\begin{itemize} 
\item As before, each epoch contains $f+1$ views. Now, though, we consider also \emph{super-epochs}, which may contain any number of epochs. Each super-epoch will guarantee the confirmation of one new block. 
\item The genesis block corresponds to view 0 of epoch 0 of super-epoch 0. 
\item We consider QCs ordered by super-epoch, then by epoch, then by view, and then by stage number within the view.
\item Replicas begin in super-epoch 1. Upon entering any super-epoch, replicas immediately wish to enter epoch 1 of that super-epoch. 
\item Within each super-epoch we rotate leaders as before, and also rotate so that the leaders of epoch 1 of super-epoch $s$ are replicas $s\text{ mod }n, \dots, s+f \text{ mod }n$.

\item Upon first seeing any confirmed block $B$ corresponding to super-epoch $s$, correct replicas broadcast $B$, together with a stage 3 QC for $B$, and then immediately begin super-epoch $s+1$. 
\end{itemize} 

\noindent (2) We stipulate how requests should be included in blocks as follows: 
\begin{itemize} 
\item A block whose parent corresponds to a previous super-epoch cannot contain requests included in any predecessor block (otherwise it fails to be valid). 
\item Blocks whose parent belongs to the same super-epoch must contain the same ordered set of requests  as their parent block to be valid. 
\end{itemize} 

\noindent (3) We change the information attached to messages as follows: 
\begin{itemize} 
\item  When a replica sends a $\mathtt{view}\ v$ message, it now attaches only the highest QC it has seen, together with the block corresponding to that QC. 
\item When a leader proposes a new block $B$ with parent block $B'$ other than the genesis block, it now broadcasts only $B,B'$ and the QC for $B'$, together with the relevant VC. Correct replicas will ignore the new block unless it arrives with $B'$, a QC for $B'$, and until they can verify that the conditions on requests included in $B$ (described above) are satisfied.  If $B'$ is the genesis block, then the leader just broadcasts $B$, together with the relevant VC. 
\end{itemize} 

\noindent (4) The rest of the instructions while within any given view or epoch are exactly as before, except that (just as replicas will not enter any view unless already in the corresponding epoch) replicas will not enter any epoch unless already in the corresponding super-epoch. Replicas will also cease carrying out the instructions for any super-epoch as soon as they leave that super-epoch.

\vspace{0.3cm} 
\noindent \textbf{Verification}. First of all, we consider the structure of blocks and the relationship with super-epochs. 

\vspace{0.2cm} 
\noindent \emph{The requests inside blocks in a chain.} Just as in the proof of Theorem \ref{t3}, it follows inductively that for any block $B$  that receives a QC, all predecessors (other than the genesis block) receive QCs, and all predecessors that belong to the same super-epoch record the same ordered set of requests. Again, the proof of Lemma \ref{consistency}  suffices to show that incompatible blocks cannot be confirmed. All confirmed blocks corresponding to the same super-epoch must therefore record the same ordered set of requests. When a correct replica first sees a block corresponding to super-epoch $s$ confirmed, it follows inductively that, for each previous super-epoch, at least one correct replica must have seen a block corresponding to that super-epoch confirmed, and will have broadcast that block together with a stage 3 QC for the block. This means that, when any correct replica first sees a block corresponding to a given super-epoch confirmed, within time $\Delta$ it will know at least one confirmed block from each previous super-epoch, and will therefore have the information it needs to execute the requests in order. 

\vspace{0.2cm} 
\noindent \emph{Liveness.} To establish liveness, we follow almost exactly the same proof as for Lemma \ref{liveness}.  
 Let $s_0$ and $e_0$ be the greatest super-epoch and epoch that any correct replica is in at any time $<\text{GST}$ (so that $e_0$ is an epoch within $s_0$).  Let $t_1$  be the first time (if there exists such) at which any correct replica enters either super-epoch $s_0+1$ (meaning that they have seen a block corresponding to $s_0$ confirmed) or epoch $e_0+1$ of  $s_0$.  Whenever we refer to epoch $e_0+1$ (or epoch $e_0$) in the following, it is always to be assumed that this means epoch $e_0+1$ (or $e_0$) of super-epoch $s_0$. Let
$v_1$ be the least view of epoch $e_0+1$ which has a correct leader.

 We produce a bound on $t_1$, by much the same argument as in the proof of  Lemma \ref{liveness}. All correct replicas will see confirmed blocks corresponding to all super-epochs $<s_0$ and an EC for epoch $e_0$ by  $GST +\Delta$, and so will either be in a super-epoch $>s_0$ or else in an  epoch $\geq e_0$ of $s_0$ by this time. The same reasoning as in the proof of Lemma \ref{liveness} then suffices to show that $t_1\leq GST +2 \Delta + 12(f+1)\Delta $. 
  Let $t^{\ast}$ be the first time (if it exists) at which a correct replica sees a block corresponding to super-epoch $s_0$  confirmed,  and let $t_2 = \text{min} \{ t^{\ast}, t_1+12(v_1+1)\Delta \}$. Note that:

\begin{itemize} 
\item[($\diamond_0)$] No correct replica can enter an epoch $>e_0+1$ prior to $t_2$. 
\item[($\diamond_1)$] No correct replica can enter any view $v>v_1$ of epoch $e_0+1$  prior to $t_2$. 
\end{itemize} 

We claim that:

\begin{enumerate} 
\item[$(\diamond_2)$]  The number of bits sent by correct replicas (combined) in the time interval $(\text{GST}+\Delta, t_2]$ is $O(\kappa n^2)$. 
\item[$(\diamond_3)$]  A block corresponding to super-epoch $s_0$ will be confirmed and seen by a correct replica by time $t_2$.
\end{enumerate} 

From $(\diamond_2)$ and $(\diamond_3)$ it follows that $O(\kappa n^2)$ bits are sent by all correct replicas combined after $GST+\Delta$ and before the first block confirmation. It also follows that the  time to first confirmation after $GST$ is  $O(f\Delta )$,  since  $t_1\leq GST +2 \Delta + 12(f+1)\Delta $ and $t_2 \leq  t_1+12(v_1+1)\Delta$.

As before, statement $(\diamond_2)$ follows straight from the definition of the protocol. 
We are left to establish  $(\diamond_3)$, which again follows essentially just as in Hotstuff. Suppose, towards a contradiction, that no block corresponding to super-epoch $s_0$ is confirmed and seen by a correct replica by  time $t_1 +12(v_1+1)\Delta$. By $(\diamond_0)$, it follows that all correct replicas will be in epoch $e_0+1$  by time $t_1+\Delta$. By $(\diamond_0)$ and $(\diamond_1)$, it follows that all correct replicas will be in view $v_1$ of epoch $e_0+1$ by time $t_1+3\Delta +12v_1\Delta$ (and within time $\Delta$ of each other). Before broadcasting the VC for view $v_1$, the leader receives $\mathtt{view}\ v_1$ messages from at least $n-2f$ correct replicas, and each of these messages includes the highest QC that the replica has seen at the point that the message is sent, together with the corresponding block. When correct replicas wish to enter view $v_1$ and send a $\mathtt{view}\ v_1$ message, they cease executing instructions for previous views. Now consider the highest lock of any correct replica $i$  before receiving the block $B$ from the leader of view $v_1$. This lock was set because $i$ saw a stage 2 QC for a block $B'$. This means that at least $n-2f$ correct replicas must have produced stage 2 votes for $B'$ (prior to any point at which they wished to enter view $v_1$ and stopped executing instructions for views $<v_1$), and so must have seen a stage 1 QC for $B'$ by that time. Since any two sets of $n-2f$ correct replicas have a correct replica in the intersection, it follows that one of the $\mathtt{view}\ v_1$ messages seen by the leader of $v_1$ must be from a correct replica who had already seen the stage 1 QC for $B'$. 
This means that all correct replicas will then produce stage 1,2 and 3 votes for the block proposed by the  leader of view  $v_1$, and this block will be confirmed before time  $t_1 +12(i+1)\Delta$, giving the required contradiction.

\section{Directions for future work} 

We close by pointing out some weaknesses in the proofs we have presented here, and asking to what extent those weaknesses are necessary.  First of all, consider the protocol described in Section \ref{protocol}, for the proof of Theorem \ref{t1}. A weakness of this protocol is that, after GST, a \emph{single} faulty replica can produce a delay of length $\Omega (f\Delta)$ in the production of confirmed blocks. For example, if the first $f$ leaders of an epoch are all correct and produce $f$ confirmed blocks in close to zero time, then a faulty leader for view $f$ of the epoch can cause a delay of length  $\Omega (f\Delta)$ simply by failing to propose any block. 
The following (slightly imprecise) question naturally arises: 

\begin{question}
Can one improve the protocol of Section \ref{protocol} so that it still establishes the claims of Theorem \ref{t1}, while also ensuring that any set of $k$ consecutive faulty leaders can only cause a delay of length $O(k\Delta)$ in the production of confirmed blocks? 
\end{question}

Next, we consider two weakness in the proof of Theorem \ref{t2}. The protocol of Section \ref{protocol} actually produced something stronger than explicitly claimed in the statement of Theorem \ref{t1}, in that the claims of  Theorem \ref{t1} are all shown to hold if we restrict attention to confirmed blocks \emph{produced by correct leaders}, e.g.\ at most $O(n^2)$ many messages are sent after GST$+\Delta$ and prior to the first confirmation of a block \emph{produced by a correct leader}. In the proof of Theorem \ref{t2}, on the other hand, a block containing a set of requests that are proposed by a correct leader will be produced whenever the first leader of a super-epoch is correct, but super-epochs for which the first leader is faulty  may confirm a set of requests that are proposed by a faulty leader (even if it is a correct leader that eventually produces a confirmed block). This means that $O(\kappa n^2)$ many bits need to be sent for each block confirmation after GST, but if $f$ consecutive super-epochs produce blocks confirming a set of requests proposed by a faulty leader, then $O(\kappa n^3)$ many bits are required to be sent before the first confirmation of a block containing requests proposed by a correct leader. 

\begin{question}
Can one modify the proof of Theorem \ref{t2} so that the result still holds when we restrict attention to blocks that confirm sets of requests proposed by correct leaders? 
\end{question}

Another obvious weakness of Theorem \ref{t2}, is that we dropped the requirement for $O(\kappa n)$ mean communication complexity: 

\begin{question}
Consider the partial synchrony model. Does there  exist a BFT SMR protocol with optimal resiliency, and worst-case communication complexity $O(\kappa n^2)$, which is also optimistically responsive, with $O(f\Delta )$ time to first confirmation after $GST$ and $O(\kappa n)$ mean communication complexity? 
\end{question}

\section*{Acknowledgements} 
The author would like to thank Ittai Abraham, Oded Naor, Kartik Nayak and Ling Ren for a number of helpful conversations.

\newcommand{\etalchar}[1]{$^{#1}$}

\end{document}